%% file: main.tex
\documentclass{llncs}
\usepackage[utf8]{inputenc}
\usepackage[T1]{fontenc}
\usepackage{amsmath}
\usepackage{amssymb}
\usepackage{mathtools}
\usepackage{bussproofs}
\usepackage{prftree}
\usepackage{url}
\usepackage{cancel}

\usepackage{rotating}
\usepackage{multirow}

\usepackage{graphicx}

\usepackage[ruled,linesnumbered]{algorithm2e}

\usepackage{floatflt}
\usepackage{times}

\usepackage{xcolor}
\definecolor{mygreen}{RGB}{0,159,57}
\definecolor{myyellow}{RGB}{255,200,0}
\definecolor{myorange}{RGB}{255,140,0}
\colorlet{myred}{red!90}
\definecolor{myblue}{RGB}{0,50,255}

\newcommand{\eql}{{\,\simeq\,}}
\newcommand{\neql}{\not\simeq}

\usepackage{ stmaryrd }

\newcommand{\vampire}{\textsc{Vampire}}
\newcommand{\spass}{\textsc{Spass}}
\newcommand{\eprover}{\textsc{E}}

\newcommand{\Sup}{\textsc{Sup}}

\newcommand{\hcancel}[1]{\text{\cancel{${#1}$}}}

\begin{document}

\title{
  Subsumption Demodulation in First-Order Theorem Proving
}

\author{
  Bernhard Gleiss\inst{1}
  \and
  Laura Kov{\'a}cs\inst{1\and 2}
  \and
  Jakob Rath\inst{1}
}

\institute{
  TU Wien, Austria
  \and
  Chalmers University of Technology, Sweden
}

\maketitle

\input{abstract}

\input{introduction.tex}

\input{preliminaries.tex}

\input{theory.tex}

\input{implementation.tex}

\input{experiments.tex}

\input{conclusion.tex}

\bibliographystyle{splncs04}
\bibliography{bib}

\end{document}

%% file: abstract.tex
\begin{abstract}
%
Motivated by applications of first-order theorem proving to software analysis,
we introduce a new inference rule, called subsumption demodulation, to improve support for
reasoning with conditional equalities in superposition-based theorem proving. 
%
We show that subsumption demodulation is a simplification rule that does not require radical changes to the underlying
superposition calculus. 
We implemented subsumption demodulation in the theorem prover \vampire{},
by extending \vampire{} with a new clause index and adapting its multi-literal matching component.
Our experiments, using the TPTP and SMT-LIB repositories, show that subsumption demodulation in \vampire{} can solve many new problems that could so far not be solved by state-of-the-art reasoners.
\end{abstract}

%% file: introduction.tex
\section{Introduction}
For the efficiency of organizing proof search during saturation-based first-order theorem proving, simplification rules are of critical importance.
Simplification rules are inference rules that do not add new formulas to the search space,
but simplify formulas by deleting (redundant) clauses from the search space.
As such, simplification rules reduce the size of the search space and are crucial in making automated reasoning efficient.

When reasoning about properties of first-order logic with equality,
one of the most common simplification rules is demodulation~\cite{kovacs2013first} for rewriting (and hence simplifying)
formulas using unit equalities $l\eql r$, where $l,r$ are terms and $\eql$ denotes equality.
As a  special case of superposition, demodulation is  implemented in first-order provers
such as \eprover~\cite{E19}, \spass~\cite{SPASS09} and \vampire~\cite{kovacs2013first}.
Recent applications of superposition-based reasoning, for example to program analysis and verification~\cite{Rapid19},
demand however new and efficient extensions of demodulation to reason about and simplify upon conditional equalities $C\rightarrow l \eql r$,
where $C$ is a first-order formula.
Such conditional equalities may, for example, encode software properties expressed in a guarded command language,
with $C$ denoting a guard (such as a loop condition) and $l \eql r$ encoding equational properties over program variables.
We illustrate the need of considering  generalized versions of demodulation in the following example.

\begin{example}\label{ex:motivating}
Consider the following formulas expressed in the first-order theory of integer linear arithmetic: 
%
%
\begin{equation}\label{eq:motivating:dem:premises}
\begin{array}{ll}
  & f(i)\eql g(i)\\
 & 0\leq i < n \rightarrow P(f(i))
\end{array}
\end{equation}
Here,  $i$ is an implicitly  universally quantified logical variable of integer sort, and $n$ is integer-valued constant. First-order reasoners will first clausify formulas~\eqref{eq:motivating:dem:premises}, deriving:
\begin{equation}\label{eq:motivating:dem:premises:cnf}
\begin{array}{ll}
  & f(i)\eql g(i)\\
 & 0 \nleq i  \lor i \nless n \lor P(f(i))
\end{array}
\end{equation}
By applying demodulation over~\eqref{eq:motivating:dem:premises:cnf},
the formula $0 \nleq i \lor i \nless n \lor P(f(i))$ is rewritten\footnote{assuming that $g$ is simpler/smaller than $f$}
using the unit equality $f(i)\eql g(i)$, yielding the clause $0 \nleq i  \lor i \nless n \lor P(g(i))$.
That is, $0 \leq i  <n \rightarrow  P(g(i))$ is derived from~\eqref{eq:motivating:dem:premises} by one application of demodulation.

Let us now consider a slightly modified version of~\eqref{eq:motivating:dem:premises}, as below: 
\begin{equation}\label{eq:motivating:premises}
\begin{array}{ll}
	 &0\leq i <n \rightarrow f(i)\eql g(i)\\
        &0\leq i < n \rightarrow P(f(i))
\end{array}
\end{equation}
whose clausal representation is given by: 
\begin{equation}\label{eq:motivating:premises:cnf}
\begin{array}{ll}
	 &0 \nleq i  \lor i \nless n \lor f(i)\eql g(i)\\
        &0 \nleq i  \lor i \nless n \lor  P(f(i))
\end{array}
\end{equation}


It is again obvious that from~\eqref{eq:motivating:premises} one can derive the formula $0\leq i <n \rightarrow P(g(i))$, or equivalently the clause:
\begin{equation}\label{eq:motivating:concl:cnf}
  0 \nleq i  \lor i \nless n \lor  P(g(i))
\end{equation}
Yet, {\it one cannot anymore apply 
demodulation-based simplification over~\eqref{eq:motivating:premises:cnf} to derive such a clause}, 
as~\eqref{eq:motivating:premises:cnf} contains no unit equality.
\qed
\end{example}
%
%

In this paper we propose a generalized version of demodulation, called {\it subsumption demodulation}, allowing 
to rewrite terms and simplify formulas using rewriting based on conditional equalities, such as in~\eqref{eq:motivating:premises}.
To do so, we extend demodulation with subsumption, that is with deciding whether (an instance of a) clause $C$ is a submultiset of a  clause $D$. This way, subsumption demodulation can be applied to non-unit clauses and is not restricted to have at least one premise clause that is a unit equality.  
We show that subsumption demodulation is a simplification rule of the superposition framework (Section~\ref{sec:sd}), allowing for example to derive the clause~\eqref{eq:motivating:concl:cnf} from~\eqref{eq:motivating:premises} in one inference step.
 By properly adjusting clause indexing and multi-literal matching in first-oder theorem provers, we provide an efficient implementation of subsumption demodulation in \vampire{} (Section~\ref{sec:implementation}) and evaluate our work against state-of-the-art reasoners, including \eprover~\cite{E19}, \spass~\cite{SPASS09}, {\sc CVC4}~\cite{CVC4} and {\sc Z3}~\cite{Z3} (Section~\ref{sec:experiments}).

\paragraph{\bf Related work.}
While several approaches generalize demodulation in superposition-based theorem proving, we argue that subsumption demodulation improves existing methods either in terms of applicability and/or efficiency.
The AVATAR architecture of first-order provers~\cite{AVATAR14}
splits general clauses into components with disjoint sets of variables, 
potentially enabling demodulation inferences whenever some of these components become unit equalities.  
Example~\ref{ex:motivating} demonstrates that subsumption demodulation solves applies in situations where AVATAR does not:
in each clause of~\eqref{eq:motivating:premises:cnf}, all literals share the variable $i$ and
hence none of the clauses from~\eqref{eq:motivating:premises:cnf} can be split using AVATAR.
That is, AVATAR would not generate unit equalities from~\eqref{eq:motivating:premises:cnf}, and therefore cannot apply demodulation over~\eqref{eq:motivating:premises:cnf} to derive~\eqref{eq:motivating:concl:cnf}. 

The local rewriting approach of~\cite{Weidenbach01} requires rewriting equality literals to be maximal\footnote{w.r.t. clause ordering} in clauses.
However, following~\cite{kovacs2013first}, for efficiency reasons we consider equality literals to be ``smaller'' than non-equality literals.
In particular, the equality literals of clauses~\eqref{eq:motivating:premises:cnf} are ``smaller'' than the non-equality literals,
preventing thus the application of {local rewriting} in Example~\ref{ex:motivating}.

We further note that the contextual rewriting rule of~\cite{BG94} is more general than our rule of subsumption demodulation.
Yet, efficiently automating contextual rewriting is extremely challenging,
while subsumption demodulation requires no radical changes in the existing machinery of superposition provers (see Section~\ref{sec:implementation}).

To the best of our knowledge, except \spass~\cite{SPASS09},
no other state-of-the-art superposition prover implements variants of conditional rewriting.
%
Subterm contextual rewriting~\cite{WeidenbachWischnewski:2008}
is a refined notion of contextual rewriting and is implemented in \spass.
A major difference of subterm contextual rewriting when compared to subsumption demodulation is that
in subsumption demodulation the discovery of the
substitution is driven by the side conditions
whereas in subterm contextual rewriting the side conditions
are evaluated by checking the validity of certain implications
by means of a reduction calculus.
This reduction calculus recursively applies another restriction of contextual rewriting
called recursive contextual ground rewriting, among other standard reduction rules.
While subterm contextual rewriting is more general, we believe that
the benefit of subsumption demodulation comes with its relatively easy and efficient integration within existing superposition reasoners,
as evidenced also in Section~\ref{sec:experiments}.

Local contextual rewriting~\cite{HillenbrandPiskacWaldmannWeidenbach:2013} is another refinement of contextual rewriting implemented in \spass.
In our experiments it performed similarly to subterm contextual rewriting.

Finally, we note that SMT-based reasoners also implement various methods to efficiently handle conditional equalities, see e.g.~\cite{CVC4Strings,Horn15}.
Yet, the setting is very different as they rely on the DPLL(T) framework~\cite{DPLLT} rather than implementing superposition.


\paragraph{\bf Contributions.}
Summarizing, this paper brings the following contributions.
\begin{itemize}
    \item To improve reasoning in the presence of conditional equalities, we introduce
        the new inference rule {\it subsumption demodulation}, which generalizes demodulation to
        non-unit equalities by combining demodulation and subsumption (Section~\ref{sec:sd}).
    \item Subsumption demodulation does not require radical changes to the underlying
        superposition calculus. We implemented subsumption demodulation in the first-order theorem prover \vampire, by extending \vampire{} with a new clause index and adapting its multi-literal matching component (Section~\ref{sec:implementation}). 
    \item
        We compared our work against state-of-the-art reasoners,
        using the TPTP and SMT-LIB benchmark repositories.
        Our experiments show that subsumption demodulation in \vampire{} can solve 11 first-order problems
        that could so far not be solved by any other state-of-the-art provers,
        including \vampire{}, \eprover, \spass{},  \textsc{CVC4} and {\sc Z3} (Section~\ref{sec:experiments}).
\end{itemize}

%% file: preliminaries.tex
\section{Preliminaries}\label{sec:prelim}
\vspace{-0.3\baselineskip}
For simplicity, in what follows we consider standard first-order logic with
equality, where equality is denoted by $\eql$.
We support all standard boolean
connectives and quantifiers in the language. Throughout the paper, we
denote terms by $l,r,s,t$, variables by $x,y$, constants by $c,d$,
function symbols by $f,g$ and predicate symbols by $P, Q, R$, all possibly with
indices. Further, we denote literals by $L$ and
clauses by $C, D$,
again possibly with indices. We
write $s\neql t$ to denote the formula $\neg s\eql t$.
A literal $s\eql t$ is called  an \emph{equality literal}.
We consider clauses as multisets of literals and denote by
$\subseteq_M$ the subset relation among multisets.
A clause that only consists of one one equality literal is called a \emph{unit equality}.

An expression $E$ is a term, literal, or clause. We write $E[s]$ to
mean an expression $E$ with a particular occurrence of a term $s$.
A \emph{substitution}, denoted by $\sigma$,  is any finite mapping of
the form $\{x_1\mapsto t_1, \ldots, x_n\mapsto t_n\}$, where $n > 0$. Applying a
substitution $\sigma$ to an expression $E$ yields another
expression, denoted by $E\sigma$, by simultaneously replacing each
$x_i$ by $t_i$ in $E$. We say that $E\sigma$ is an instance of $E$. A \emph{unifier} of two expressions $E_1$ and
$E_2$ is a substitution $\sigma$ such that $E_1\sigma =
E_2\sigma$. If two expressions have a unifier,  they also have a
\emph{most general unifier (mgu)}.
A \emph{match} of expression $E_1$ to expression $E_2$
is a substitution $\sigma$ such that $E_1\sigma = E_2$.
Note that any match is a unifier (assuming the sets of variables in $E_1$ and $E_2$ are disjoint),
but not vice-versa, as illustrated below.

\begin{example}\label{ex:match:unif}
Let $E_1$ and $E_2$ be the clauses $Q(x,y)\lor R(x,y)$
and $Q(c,d)\lor R(c,z)$, respectively.
The only possible match of $Q(x,y)$ to $Q(c,d)$ is $\sigma_1=\{x\mapsto c, y\mapsto d\}$.
On the other hand,
the only possible match of $R(x,y)$ to $R(c,z)$ is $\sigma_2=\{x\mapsto c, y\mapsto z\}$.
As $\sigma_1$ and $\sigma_2$ are not the same, there is no match of $E_1$ to $E_2$.
Note however that $E_1$ and $E_2$ can be unified;
for example, using $\sigma_3=\{x\mapsto c, y\mapsto d, z\mapsto d\}$.
\end{example}

\paragraph{\bf Superposition inference system.}
We assume basic knowledge in first-order theorem proving and
superposition reasoning~\cite{Ganzinger01,Rubio01}. 
We adopt the notations and the inference system of superposition
from~\cite{kovacs2013first}. We recall that first-order provers perform inferences
on clauses using inference rules, where  an \emph{inference} is usually written as: 
\vspace{-0.5em}
\begin{prooftree}
    \AxiomC{$C_1\qquad \ldots$}
    \AxiomC{$C_n$}
    \BinaryInfC{$C$}
\end{prooftree}
 with $n\geq 0$. The clauses $C_1, \ldots, C_n$ are called the
 premises and $C$ is the conclusion of the inference above.
 An inference is \emph{sound} if its conclusion is a
logical consequence of its premises.
An inference rule is a set of inferences and an inference system is a
set of inference rules. An inference system is \emph{sound} if all its
inference rules are sound.

Modern  first-order theorem provers implement the
\emph{superposition inference system} for first-order logic with
equality. This inference system is parametrized by
a \emph{simplification ordering} over
terms and a \emph{literal selection function} over clauses.
In what follows, we denote by  $\succ$ a simplification ordering over
terms, that is  $\succ$ is a {well-founded partial ordering} satisfying
the following three conditions:
\vspace{-0.7\baselineskip}
\begin{itemize}
\item \emph{stability under substitutions}: if $s\succ t$, then $s\theta \succ t\theta$;
\item \emph{monotonicity}: if $s\succ t$, then $l[s] \succ l[t]$;
\item \emph{subterm property}:   $s\succ t$ whenever $t$ is a proper subterm of $s$.
\end{itemize}
 The simplification ordering $\succ$ on terms can be extended to a simplification
 ordering on literals and clauses, using a multiset extension of orderings. For simplicity, the
 extension of $\succ$ to literals and clauses will also be denoted by
 $\succ$. Whenever $E_1 \succ E_2$, we say that $E_1$ is bigger than
 $E_2$ and $E_2$ is smaller than $E_1$ w.r.t. $\succ$.
We say that an equality literal $s\eql t$ is \emph{oriented}, if
$s\succ t$ or $t\succ s$. 
The literal extension of  $\succ$ asserts
 that negative literals are always bigger than their positive
 counterparts. Moreover, if $L_1\succ L_2$, where $L_1$ and $L_2$ are
 positive, then $\neg L_1\succ L_1\succ  \neg L_2 \succ L_2$.
Finally,  equality literals are set to be smaller than
any literal using a predicate different than $\eql$.
 
A \emph{selection function} selects at least one literal in every
non-empty clause. In what follows, selected literals in clauses will
be underlined: when writing $\underline{L}\lor C$, we mean that (at
least) $L$ is selected in $L\lor C$. In what follows, we assume that
selection functions are \emph{well-behaved} w.r.t. $\succ$: either a
negative literal is selected or
all maximal literals w.r.t. $\succ$ are selected.

In the sequel, we fix a simplification ordering $\succ$ and a
well-behaved selection function and consider the superposition
inference system, denoted by $\Sup$,  parametrized by these
two ingredients. The inference system $\Sup$ for first-order
logic with equality consists of the
inference rules of Figure~\ref{fig:sup}, and it is both sound and
refutationally complete. That is, if a set $S$ of clauses is
unsatisfiable, then the empty clause (that is, the always false
formula) is derivable from $S$ in $\Sup$.

\begin{figure}[t]
  {\small 
  \begin{itemize}
  \item Resolution and Factoring
    \[
    \AxiomC{$\underline{L} \lor C_1$}
    \AxiomC{$\underline{\lnot L'} \lor C_2$}
    \BinaryInfC{$(C_1\lor C_2)\sigma$}
    \DisplayProof
    \hspace{6em}
    \AxiomC{$\underline{L} \lor \underline{L'} \lor C$}
    \UnaryInfC{$(L \lor C)\sigma$}
    \DisplayProof
    \]
    where $L$ is not an equality literal and $\sigma = mgu(L, L')$
    \\

  \item Superposition
    \[
    \AxiomC{$\underline{s \eql t} \lor C_1$}
    \AxiomC{$\underline{L \lbrack s' \rbrack} \lor C_2$}
    \BinaryInfC{$(C_1 \lor L \lbrack t \rbrack \lor C_2) \theta$}
    \DisplayProof\qquad\qquad\qquad\qquad\qquad\qquad\qquad\qquad\quad
  \]
  
  \[
    \AxiomC{$\underline{s \eql t} \lor C_1$}
    \AxiomC{$\underline{l \lbrack s' \rbrack \eql l'} \lor C_2$}
    \BinaryInfC{$(C_1 \lor l \lbrack t \rbrack \eql l' \lor C_2) \theta$}
    \DisplayProof
    \hspace{3em}
    \AxiomC{$\underline{s \eql t} \lor C_1$}
    \AxiomC{$\underline{l \lbrack s' \rbrack \neql l'} \lor C_2$}
    \BinaryInfC{$(C_1 \lor l \lbrack t \rbrack \neql l' \lor C_2) \theta$}
    \DisplayProof
  \]
    where $s'$ not a variable, $L$ is not an equality,
    $\theta = mgu(s, s')$, $t \theta \not\succ s \theta$ and $
    l'\theta\not\succ l \lbrack
    s' \rbrack \theta$ \\

  \item Equality Resolution and Equality Factoring
    \[
    \AxiomC{$\underline{s \neql s'} \lor C$}
    \UnaryInfC{$C \theta$}
    \DisplayProof
    \hspace{6em}
    \AxiomC{$s \eql t \lor \underline{s' \eql t'} \lor C$}
    \UnaryInfC{$(s \eql t \lor t \neql t' \lor C)\theta$}
    \DisplayProof
    \]
    where $\theta = mgu(s, s')$,
    $t\theta\not\succ s\theta$ and $t'\theta\not\succ t\theta$
  \end{itemize}}
  \caption{The superposition calculus $\Sup$\label{fig:sup}.}
\end{figure}

\section{Superposition-based Proof Search}\label{sec:sup}
We now overview the main ingredients in organizing proof 
search within first-order provers, using  the
superposition calculus. For details, we refer
to~\cite{Ganzinger01,Rubio01,kovacs2013first}. 

Superposition-based provers use \emph{saturation algorithms}: applying
all possible inferences of $\Sup$ in a certain order to the clauses in the
search space until (i) no more inferences can be applied or (ii) the
empty clause has been derived. A simple implementation of a saturation
algorithm would however be very inefficient as applications of all
possible inferences will quickly blow up the search space.

Saturation
algorithms can however be made efficient by exploiting a powerful
concept of \emph{redundancy}: deleting so-called redundant clauses from the
search space by  preserving completeness of $\Sup$. A clause $C$ in a
set $S$ of clauses (i.e. in the search space) is \emph{redundant} in
$S$, if there exist clauses $C_1,\dots, C_n$ in $S$,
such that $C\succ C_i$  and $C_1,\dots, C_n \vDash C$. That is, a
clause $C$ is redundant in $S$ if it is a logical consequence of
clauses that are smaller than $C$ w.r.t. $\succ$.
It is known that redundant clause can be removed from the search space
without affecting completeness of superposition-based proof
search. For this reason, saturation-based theorem provers, such as \eprover{},
\spass{} and \vampire{}, not only generate new clauses but also delete
redundant clauses during proof search by using both \emph{generating} and
\emph{simplifying} inferences.

\noindent{\bf Simplification rules.}  A \emph{simplifying inference}
is an inference in which one premise $C_i$ becomes redundant after the
addition of the conclusion $C$ to the search space, and hence $C_i$
can be deleted.
In what follows, we will denote deleted clauses by
drawing a line through it and refer to simplifying inferences as
\emph{simplification rules}. 
The premise $C_i$ that becomes redundant is called
the \emph{main premise}, whereas other premises are
called \emph{side premises} of the simplification rule. Intuitively, 
a simplification rule simplifies its main premise to its conclusion by using
additional knowledge from its side premises. 
Inferences that are not simplifying are
called \emph{generating}, as they generate and add a new clause $C$
to the search space.

In saturation-based proof search,  we
distinguish between \emph{forward} and \emph{backward}
simplifications. 
During forward simplification, a newly derived clause is
simplified using previously derived clauses as side clauses. Conversely, during backward
simplification a newly derived clause is used as side clause to simplify previously derived
clauses.




\noindent{\bf Demodulation.} One example of a simplification rule is \emph{demodulation}, or also called
\emph{rewriting by unit equalities}. Demodulation is the following
inference rule: 

\begin{prooftree}
		\AxiomC{$l \eql r$}
		\AxiomC{{$\cancel{L[t] \lor C}$}}
		\BinaryInfC{$L[r\sigma]\lor C$}
	\end{prooftree}
where $l\sigma=t$,  $l\sigma \succ r\sigma$ and $L[t]\lor C \succ
(l\eql r)\sigma$, for some substitution $\sigma$. 
        
It is easy to see that demodulation is a simplification rule. 
Moreover, demodulation is 
special case of a superposition inference where one premise of the
inference is deleted. However, unlike a superposition inference,
demodulation is not restricted to selected literals.

\begin{example}
    \newcommand{\cside}{   f(f(x)) \eql f(x)  }
    \newcommand{\cmain}{ P(f(f(c))) \lor Q(d) }
    \newcommand{\cconc}{ P(  f(c) ) \lor Q(d) }
    Consider the clauses $C_1 = \cside$ and $C_2 = \cmain$.
    Let $\sigma$ be the substitution $\sigma = \{ x \mapsto c \}$.
    By the subterm property of $\succ$, we have $f(f(c))\succ
    f(c)$. Further, as equality literals are smaller than non-equality
    literals, we have $ P(f(f(c))) \lor Q(d) \succ f(f(c)) \eql f(c)$.
    We thus apply demodulation and 
    $C_2$ is simplified into the clause~$C_3 = \cconc$:
    \[
        \prftree
        { \cside \quad }
        { \hcancel{P({f(f(c))}) \lor Q(d)} }
        {          P({  f(c) }) \lor Q(d)  }
        \tag*{$\qed$}  
    \]
\end{example}



\noindent{\bf Deletion rules.}
Even when simplification rules are in use,
deleting more/other redundant clauses  is still useful to keep the
search space small. For this reason, in addition to simplifying and
generating rules, theorem provers also use \emph{deletion rules}:
a \emph{deletion rule} checks whether clauses in the search space are
redundant due to the presence of other clauses in the search space,
and removes redundant clauses from the
search space.

%
Given clauses $C$ and $D$, we say $C$ subsumes $D$ if there is some
substitution $\sigma$ such that $C\sigma$ is a submultiset of $D$,
that is $C\sigma \subseteq_M D$. 
\emph{Subsumption} is the deletion rule that removes subsumed clauses from the
search space.

\begin{example}
  Let $C = P(x) \lor Q(f(x))$
and  $D = P(f(c)) \lor {P(g(c))} \lor Q(f(c)) \lor {Q(f(g(c)))} \lor
R(y)$ be clauses in the search space. Using  $\sigma=\{ x \mapsto g(c) \}$, it is
easy to see that $C$ subsumes $D$, and hence $D$ is deleted from the
search space.
    \qed
\end{example}




%% file: theory.tex
\section{Subsumption Demodulation}\label{sec:sd}

In this section we introduce a new simplification rule, called
subsumption demodulation, by extending demodulation to a
simplification rule over conditional equalities.
We do so by combining demodulation with subsumption checks to find
simplifying applications of rewriting by non-unit (and hence
conditional) equalities.

\subsection{Subsumption Demodulation for Conditional Rewriting}\label{sec:sd:sound}

Our rule of subsumption demodulation is defined below.
\begin{definition}[Subsumption Demodulation]
    \label{def:subsumption-demodulation}
    \emph{Subsumption demodulation}
    is the inference rule:
    \begin{equation}%
        \label{eq:subsumption-demodulation}
        \prftree
        {l \eql r \lor C \quad}
        {{L[t] \lor D}}
        {L[r\sigma] \lor D}
    \end{equation}
    where:
    \begin{enumerate}
        \item
            $l \sigma = t$,
        \item
            $C\sigma \subseteq_M D$,
        \item
            $l\sigma \succ r\sigma$, and
        \item%
            \label{item:side-condition-main-gt-side}
            $L[t] \lor D \succ (l \eql r)\sigma \lor C\sigma$.
    \end{enumerate}

    We call the equality $l \eql r$ in the left premise of~\eqref{eq:subsumption-demodulation}
    the \emph{rewriting equality} of subsumption demodulation.
\end{definition}

It is easy to see that if $l \eql r \lor C$ and $L[t] \lor D$ are valid,
then $L[r\sigma] \lor D$ also holds.
We thus conclude:

\begin{theorem}[Soundness]%
    \label{thm:sd-soundness}
    Subsumption demodulation is sound.
\end{theorem}

Detecting possible applications of subsumption demodulation involves
(i) selecting one equality of the side clause as rewriting equality and
(ii) matching each of the remaining literals,
denoted $C$ in~\eqref{eq:subsumption-demodulation},
to some literal in the main clause.
Step (i) is similar to finding unit equalities in demodulation, whereas
step (ii) reduces to showing that $C$ subsumes parts of the main premise.
Informally speaking, subsumption demodulation combines
demodulation and subsumption, as discussed in Section~\ref{sec:implementation}.
Note that in step (ii), matching allows any instantiation of $C$ to $C\sigma$
via substitution~$\sigma$;
yet, we we do \emph{not} unify the side and main premises of subsumption demodulation,
as illustrated later in Example~\ref{ex:sd-no-unification}.
Furthermore,
we need to find a term $t$ in the unmatched part $D \setminus C\sigma$ of the main premise,
such that $t$ can be rewritten according to the rewriting equality into $r\sigma$.

As the ordering $\succ$ is partial, the conditions of
Definition~\ref{def:subsumption-demodulation}  must be checked a
posteriori, that is after subsumption demodulation has been applied
with a fixed substitution and revise the substitution if needed.
Note however that if $l\succ r$ in the rewriting equality, then
$l\sigma\succ r\sigma$ for any substitution, so checking the ordering a
priori helps, as illustrated in the following example.

\newcommand{\emphDemTerm}[1]{{#1}}
\newcommand{\emphSubsLit}[1]{{#1}}

\begin{example}%
    \label{ex:sd}
    \newcommand{\cside}{   f(g(x)) \eql g(x) \lor Q(x)           \lor R(y)       }
    \newcommand{\cmain}{ P(f(g(c)))          \lor Q(c) \lor Q(d) \lor R(f(g(d))) }
    \newcommand{\cconc}{ P(  g(c) )          \lor Q(c) \lor Q(d) \lor R(f(g(d))) }
    Let us consider the following two clauses:
    \begin{align*}
        C_1 &= \cside 
        \\
        C_2 &= \cmain
    \end{align*}
    By the subterm property of~$\succ$, we conclude that $f(g(x))\succ g(x)$.
    Hence, the rewriting equality, as well as any instance of it, is oriented.

    Let $\sigma$ be the substitution
    $\sigma = \{ x \mapsto c, y \mapsto f(g(d)) \}$.
    Due to the previous paragraph, we know $f(g(c)) \succ g(c)$
    As equality literals are smaller than non-equality ones,
    we also conclude $P(f(g(c))) \succ f(g(c)) \eql g(c)$.  Thus, we
    have
    $P(f(g(c)))          \lor Q(c) \lor Q(d) \lor R(f(g(d))) ~\succ ~
    f(g(c)) \eql g(c) \lor Q(c)           \lor R(f(g(d)))$ and we can
    apply subsumption demodulation to $C_1$ and $C_2$,
    deriving clause $C_3 = \cconc$.


    We note that demodulation cannot derive $C_3$ from $C_1$ and $C_2$, as there is no unit equality.
    \qed
\end{example}

Example~\ref{ex:sd} highlights limitations of demodulation
when compared to subsumption demodulation.
We next illustrate different possible applications of subsumption
demodulation using a fixed side premise and different main premises.


\begin{example}%
    \label{ex:sd-not-pre-oriented}

    Consider the clause~$C_1 = f(g(x)) \eql g(y) \lor Q(x) \lor R(y)$.
    Only the first literal~$f(g(x)) \eql g(y)$ is a positive equality
    and as such eligible as rewriting equality.
    Note that~$f(g(x))$ and~$g(y)$ are incomparable w.r.t. $\succ$ due
    to occurrences of different variables,
    and hence whether $f(g(x))\sigma\succ g(y)\sigma$ depends on the
    chosen substitution $\sigma$.

    \noindent (1)
            Consider the
            clause~$C_2 = P({f(g(c))}) \lor {Q(c)} \lor {R(c)}$
            as the main premise.
            With the substitution $\sigma_1=\{ x \mapsto c, y \mapsto
            c \}$, 
            we have~$f(g(x))\sigma_1 \succ g(x)\sigma_1$ as
            $f(g(c))\succ g(c)$ due to the
            subterm property of $\succ$, 
            enabling a possible application of subsumption
            demodulation over $C_1$ and $C_2$.

\noindent (2)
            Consider now~$C_3 = P(\emphDemTerm{g(f(g(c)))}) \lor
            \emphSubsLit{Q(c)} \lor \emphSubsLit{R(f(g(c)))}$ as the
            main premise 
            and the substitution $\sigma_2=\{ x \mapsto c, y \mapsto
            f(g(c)) \}$.
            We have $g(y)\sigma_2\succ f(g(x))\sigma_2$,
            as $g(f(g(c)) \succ f(g(c))$. The instance of the
            rewriting
            equality is oriented differently in this case than in the previous
            one, enabling a 
            possible application of subsumption
            demodulation over $C_1$ and $C_3$. 

            \noindent (3)
            On the other hand,
            using the clause~$C_4 = P(f(g(c))) \lor \emphSubsLit{Q(c)} \lor \emphSubsLit{R(z)}$
            as the main premise, the only substitution we can use is $\sigma_3=\{ x \mapsto c, y \mapsto z \}$.
            The corresponding instance of the rewriting equality is then
            $f(g(c)) \eql g(z)$, which cannot be oriented in general.
            Hence, subsumption demodulation cannot be
            applied in this case,
            even though we can find the matching term $f(g(c))$ in $C_4$.
            \qed
          \end{example}

As mentioned before, the substitution $\sigma$ appearing in subsumption demodulation can only be used
to instantiate the side premise, but not for unifying side and main premises, as we would not obtain a simplification rule.

\begin{example}%
    \label{ex:sd-no-unification}
    Consider the clauses:
    \begin{align*}
        C_1 &= f(c) \eql c \lor Q(d)
        \\
        C_2 &= P(f(c)) \lor Q(x)
    \end{align*}
    As we cannot match $Q(d)$ to $Q(x)$ (although we could match $Q(x)$ to $Q(d)$),
    subsumption demodulation is not applicable with premises $C_1$ and $C_2$.
    \qed
\end{example}


\subsection{Simplification using Subsumption Demodulation}\label{sec:sd:simplif}
Note that in the special case where $C$ is the empty clause
in~\eqref{eq:subsumption-demodulation},
subsumption demodulation reduces to demodulation and hence it is a
simplification rule.
We next show that this is the case in general:

\begin{theorem}[Simplification rule]%
    \label{thm:sd-simplifying}
    Subsumption demodulation is a simplification rule and we have:
    \begin{equation*}%
        \prftree
        {l \eql r \lor C \quad}
        {\hcancel{L[t]       \lor D}}
        {         L[r\sigma] \lor D }
    \end{equation*}
    where:
    \begin{enumerate}
        \item
            $l\sigma = t$,
        \item%
            \label{item:CsigmaSubseteqD}
            $C\sigma \subseteq_M D$,
        \item
            $l\sigma \succ r\sigma$, and
        \item
            $L[t] \lor D \succ (l \eql r)\sigma \lor C\sigma$.
    \end{enumerate}
\end{theorem}

\begin{proof}
    Because of the second condition of the definition of subsumption demodulation,
    $L[t] \lor D$ is clearly a logical consequence of $L[r\sigma] \lor D$ and $l \eql r \lor C$.
    Moreover, from the fourth condition, 
    we trivially have  $L[t] \lor D \succ (l \eql r)\sigma \lor C\sigma$.
    It thus remains to show that $L[r\sigma] \lor D$ is smaller than $L[t] \lor D$ w.r.t. $\succ$.
    As $t = l\sigma \succ r\sigma$, the monotonicity property
    of $\succ$ asserts that $L[t] \succ L[r\sigma]$,
    and hence $L[t] \lor D \succ L[r\sigma] \lor D$.
    This concludes that $L[t] \lor D$ is redundant w.r.t. the conclusion and left-most premise of subsumption demodulation.
    \qed
\end{proof}

\begin{example}
By revisiting Example~\ref{ex:sd}, Theorem~\ref{thm:sd-simplifying}
asserts that clause $C_2$ is simplified into $C_3$, and subsumption
demodulation deletes $C_2$ from the search space. \qed
\end{example}

\subsection{Refining Redundancy}\label{sec:sd:red}

The fourth condition defining subsumption demodulation in
Definition~\ref{def:subsumption-demodulation} is
needed to ensure that the main premise of subsumption demodulation
becomes redundant. However, comparing clauses w.r.t. the ordering
$\succ$ is computationally expensive; yet, not necessary for
subsumption demodulation.
Following the notation of Definition~\ref{def:subsumption-demodulation},
let $D'$ such that $D = C\sigma \lor D'$.
%
By properties of multiset orderings, the condition
$L[t] \lor D  \succ (l \eql r)\sigma \lor C\sigma$ 
is equivalent to
$L[t] \lor D' \succ (l \eql r)\sigma$,
as the literals in $C\sigma$ occur on both sides of $\succ$.
This means, to ensure the redundancy of
the main premise of subsumption demodulation, we only need to ensure
that there is a literal from $L[t] \lor D$ such that this
literal is bigger that the rewriting equality.
\begin{theorem}[Refining redundancy]\label{thm:red}
  The following two conditions are equivalent:
  \begin{itemize}
    \item[(R1)] $L[t] \lor D  \succ (l \eql r)\sigma \lor C\sigma$
    \item[(R2)] $L[t] \lor D' \succ (l \eql r)\sigma$
  \end{itemize}
\end{theorem}

As mentioned in Section~\ref{sec:sd:sound},
application of subsumption demodulation involves checking that an ordering condition between premises holds
(side condition \ref{item:side-condition-main-gt-side} in Definition~\ref{def:subsumption-demodulation}).
Theorem~\ref{thm:red} asserts that we only need
to find a literal in $L[t] \lor D'$ that is bigger than the rewriting equality  
in order to ensure that the ordering condition is fulfilled.
In the next section we show that by re-using and properly changing
the underlying machinery of first-order  provers for demodulation
and subsumption, subsumption demodulation can efficiently be
implemented in superposition-based proof search.

%% file: implementation.tex
\setlength{\textfloatsep}{0.5\textfloatsep}
\setlength{\floatsep}{0.5\floatsep}

\section{Subsumption Demodulation in Vampire}\label{sec:implementation}

We implemented subsumption demodulation
in the first-order theorem prover \vampire{}.
Our implementation consists of about 5000 lines
of C++ code and is available at:\\[-.5em]

\noindent{\small
  \url{https://github.com/vprover/vampire/tree/subsumption-demodulation}}\\[-.5em]

As for any simplification rule, we implemented the forward and
backward versions of subsumption demodulation separately.
Our new \vampire{} options controlling subsumption demodulation are
\texttt{fsd} and \texttt{bsd}, both with possible values \texttt{on}
and \texttt{off}, to respectively enable forward and backward subsumption demodulation.

As discussed in Section~\ref{sec:sd}, subsumption demodulation
uses reasoning based on a combination of demodulation and subsumption.
Algorithm~\ref{algo:fsd} details our implementation for {\it forward
subsumption demodulation}.
In a nutshell, given a clause $D$ as main premise, (forward) subsumption
demodulation in \vampire{} consists of the following main steps:

\begin{enumerate}
    \item
        \emph{Retrieve candidate clauses} $C$ as side premises of subsumption
        demodulation (line~\ref{algo:fsd:retrieve} of Algorithm~\ref{algo:fsd}).
        To this end, we design a new clause index with imperfect filtering, by
        modifying the subsumption index in \vampire{}, as discussed later in
        this section.
    \item
        \emph{Prune candidate clauses}
        by checking the conditions of subsumption demodulation
        (lines~\ref{algo:fsd:mlmatch}--\ref{algo:fsd:red} of Algorithm~\ref{algo:fsd}),
        in particular selecting a rewriting equality 
        and matching the remaining literals of the side premise
        to literals of the main premise.
        After this, prune further by
        performing a posteriori checks for orienting the rewriting equality $E$,
        and checking the redundancy of the given main premise $D$.
        To do so, we revised multi-literal matching
        and redundancy checking in \vampire{} (see later).
    \item
        \emph{Build simplified clause} by simplifying and deleting the (main)
        premise $D$ of subsumption demodulation using (forward) simplification
        (line~\ref{algo:fsd:simpl} of Algorithm~\ref{algo:fsd}).
\end{enumerate}

Our implementation of \emph{backward subsumption demodulation} requires only a
few changes to Algorithm~\ref{algo:fsd}:
(i) we use the input clause as side premise $C$ of backward subsumption demodulation and
(ii) we retrieve candidate clauses $D$ as potential main premises of subsumption demodulation.
Additionally,
(iii) instead of returning a single simplified clause $D'$,
we record a replacement clause for each candidate clause $D$
where a simplification was possible.

\begin{algorithm}[t]
    \caption{Forward Subsumption Demodulation -- FSD}
    \label{algo:fsd}
    \DontPrintSemicolon
    \SetKwInOut{Input}{Input}
    \SetKwInOut{Output}{Output}
    \SetKwBlock{Loop}{loop}{end}
    \SetKwFor{Foreach}{for each}{do}{end}

    \Input{Clause $D$, to be used as main premise}
    \Output{Simplified clause $D'$ if (forward) subsumption demodulation is possible}

    \tcp{Retrieve candidate side premises}
    $\mathit{candidates} \coloneqq \mathit{FSDIndex.Retrieve}(D)$ \;\label{algo:fsd:retrieve}
    \Foreach{$C \in \mathit{candidates}$}{
        \While{$m = \mathit{FindNextMLMatch}(C,D)$}{ \label{algo:fsd:mlmatch}
            $\sigma' \coloneqq m.\mathit{GetSubstitution}()$ \;\label{algo:fsd:subst}
            $E \coloneqq m.\mathit{GetRewritingEquality}()$ \; \label{algo:fsd:eq}
            \tcp{$E$ is of the form $l \eql r$, for some terms $l,r$}
            \If{exists term $t$ in $D \setminus C\sigma'$ and substitution
                $\sigma \supseteq \sigma'$ such that $t=l\sigma$}
              {%
                  \If{$\mathit{CheckOrderingConditions}(D, E, t, \sigma) $}{\label{algo:fsd:red}%
                    $D' \coloneqq \mathit{BuildSimplifiedClause}(D, E, t, \sigma)$\; \label{algo:fsd:simpl} 
                    \Return $D'$
                }%
            }%
        }%
    }%
\end{algorithm}

\paragraph{\bf Clause indexing for subsumption demodulation.}
We build upon the indexing approach~\cite{SekarRamakrishnanVoronkov:2001:TermIndexing}
used for subsumption in \vampire{}:
the subsumption index in \vampire{} stores and retrieves candidate clauses for subsumption.
Each clause is indexed by exactly one of its literals.
In principle, any literal of the clause can be chosen.
In order to reduce the number of retrieved candidates,
the best literal is chosen in the sense that the chosen literal
maximizes a certain heuristic (e.g. maximal weight).  
Since the subsumption index is not a perfect index (i.e., it may retrieve non-subsumed
clauses), additional checks on the retrieved clauses are performed.

Using the subsumption index of \vampire{} as the clause index for
forward subsumption demodulation would however omit retrieving clauses
(side premises) in which the rewriting equality is chosen as key for the index,
omitting this way a possible application of subsumption demodulation.
Hence, we need a new clause index in which the best literal
can be adjusted to be the rewriting equality.
To address this issue, we added a new clause index, called the
\emph{forward subsumption demodulation index (FSD index)},
to \vampire{}, as follows:
we index potential side premises either by their
best literal (according to the heuristic), the second best literal, or both.
If the best literal in a clause $C$ is a positive equality
(i.e. a candidate rewriting equality)
but the second best is not, $C$ is indexed by the second best literal, and vice versa.
If both the best and second best literal are positive equalities, $C$ is indexed by both of them.
Furthermore, because the FSD index is exclusively used by forward subsumption demodulation,
this index only needs to keep track of clauses that contain at least one positive equality.

In the backward case, we can in fact reuse \vampire{}'s index for backward subsumption.
Instead we need to query the index by the best literal, the second best literal, or both (as described in the previous paragraph).

\paragraph{\bf Multi-literal matching.}
Similarly to the subsumption index,
our new subsumption demodulation index is  not a perfect index, that
is it performs imperfect filtering for retrieving clauses. Therefore,
additional post-checks are required on the retrieved clauses. In our
work, we devised a multi-literal matching approach to:

\noindent--
choose the rewriting equality among the literals of the side premise $C$,
and

\noindent--
check whether the remaining literals of $C$ can be uniformly
instantiated to the literals of the main premise $D$ of
subsumption demodulation.

There are multiple ways to organize this process.
A simple approach is to
(i) first pick any equality of a side premise $C$
as the rewriting equality of subsumption demodulation,
and then
(ii) invoke the existing multi-literal matching machinery of
\vampire{} to match the remaining literals of $C$ with a subset of
literals of $D$.
For the latter step (ii), the task is to find a
substitution $\sigma$ such that $C\sigma$ becomes a 
submultiset of the given clause $D$.
If the choice of the rewriting equality in step (i) turns out to be wrong, we backtrack.
In our work, we revised the existing multi-literal matching machinery
of \vampire{} to a new multi-literal matching approach for subsumption
demodulation, by using the steps (i)-(ii) and interleaving equality selection with matching.

We note that the  substitution $\sigma$ in step (ii) above is built in
two stages: first we get a partial substitution $\sigma'$ from
multi-literal matching and then (possibly) extend $\sigma'$ to
$\sigma$ by matching term instances of the rewriting equality with
terms of $D\setminus C\sigma$.

\begin{example}\label{eq:mml}
    Let $D$ be the clause $P(f(c,d)) \lor Q(c)$.
    Assume that our (FSD) clause index retrieves
    the clause $C=f(x,y) \eql y \lor Q(x)$ from the search space
    (line~\ref{algo:fsd:retrieve} of Algorithm~\ref{algo:fsd}).
    We then invoke our multi-literal matcher
    (line~\ref{algo:fsd:mlmatch} of Algorithm~\ref{algo:fsd}),
    which matches the literal $Q(x)$ of $C$ to the literal $Q(c)$ of $D$ and
    selects the equality literal~$f(x,y) \eql y$ of $C$
    as the rewriting equality for subsumption demodulation
    over $C$ and $D$.
    The matcher returns the choice of rewriting equality
    and the partial substitution~$\sigma' = \{ x \mapsto c \}$.
    We arrive at the final substitution~$\sigma = \{ x \mapsto c, y \mapsto d \}$
    only when we match the instance $f(x,y)\sigma'$, that is $f(c,y)$,
    of the left-hand side of the rewriting equality to the literal~$f(c,d)$ of $D$.
    Using $\sigma$, subsumption demodulation over $C$ and $D$ will
    derive $P(d) \lor Q(c)$,
    after ensuring that $D$ becomes redundant (line~\ref{algo:fsd:simpl} of Algorithm~\ref{algo:fsd}).
    \qed
\end{example}

We further note that multi-literal matching is an NP-complete problem.
Our multi-literal matching problems may have more than one solution,
with possibly only some (or none) of them
leading to successful applications of subsumption demodulation.
In our implementation, we examine all solutions retrieved by multi-literal matching.
We also experimented with limiting the number of matches examined after multi-literal matching
but did not observe relevant improvements.
Yet, our implementation in \vampire{} also supports an additional option
allowing the user to specify an upper bound on how many solutions of
multi-literal matching should be examined.

\paragraph{\bf Redundancy checking.}
To ensure redundancy of the main premise $D$ after the subsumption
demodulation inference,  we need to check two properties.
First, the instance $E\sigma$ of the rewriting equality $E$ must be oriented. This is a simple ordering check.
Second, the main premise $D$ must be larger than the side premise $C$. Thanks
to Theorem~\ref{thm:red}, this latter condition is  reduced to finding
a literal among the unmatched part of the main premise $D$ that is bigger than 
the instance $E\sigma$ of the rewriting equality $E$.

\begin{example}\label{eq:mml:red}
  In case of Example~\ref{eq:mml}, the rewriting equality $E$ is
  oriented and hence $E\sigma$ is also oriented. Next, the literal
  $P(f(c,d))$ is bigger than $E\sigma$, and hence $D$ is redundant
  w.r.t. $C$ and $D'$. \qed
 \end{example}

 

%% file: experiments.tex
\vspace{-1em}
\section{Experiments}\label{sec:experiments}

We evaluated our implementation of subsumption demodulation in
\vampire{} on the examples of the TPTP~\cite{Sutcliffe:2017:TPTP} and
SMT-LIB~\cite{BarFT-SMTLIB} repositories.
All our experiments were carried out on the StarExec cluster~\cite{StumpSutcliffeTinelli:2014:StarExec}.

\noindent{\bf Benchmark setup.}  From the 22,686 problems in the TPTP benchmark set, \vampire{} can
parse 18,232 problems.
Out of these problems, we only used those problems that involve
equalities as subsumption demodulation is only applicable in the
presence of (at least one) equality. As such, we used 13,924 TPTP problems
in our experiments.

On the other hand, when using the SMT-LIB repository, we chose the benchmarks from categories LIA, UF, UFDT, UFDTLIA, and
UFLIA, as these benchmarks involve reasoning with both theories and
quantifiers and the background theories are the theories that
\vampire{} supports.
These are 22,951 SMT-LIB problems in total, of which 22,833 problems remain
after removing those where equality does not occur.

\noindent{\bf Comparative experiments with \vampire{}.} As a first
experimental study, we compared the performance of subsumption
demodulation in \vampire{} for different values of \texttt{fsd} and
\texttt{bsd}, that is by using forward (FSD) and/or backward (BSD) subsumption
demodulation.
To this end, we evaluated subsumption demodulation
using the CASC and SMTCOMP schedules of \vampire{}'s portfolio mode.
In order to test subsumption demodulation with the portfolio mode,
we added the options \texttt{fsd} and/or \texttt{bsd} to \emph{all}
strategies of \vampire{}. 
While the resulting strategy schedules 
could potentially be further improved,
it allowed us to test FSD/BSD with a variety of strategies. 


\begin{table}[t]
    \centering
    \caption{%
        Comparing \vampire{} with and without subsumption demodulation
        on TPTP, using \vampire{} in portfolio mode.
    }
    \label{tab:experiment-casc}
    \begin{tabular}{l||c|c|c}
        Configuration   & Total     & Solved    & New (SAT+UNSAT)  \\ 
        \hline
        \vampire{}          & 13,924    & 9,923     & --               \\ 
        \hline
        \vampire{}, with FSD             & 13,924    & 9,757     & 20 (3+17)        \\ 
        \vampire{}, with BSD             & 13,924    & 9,797     & 14 (2+12)        \\ 
        \vampire{}, with FSD and BSD         & 13,924    & 9,734     & 30 (6+24)        \\ 
    \end{tabular}
\end{table}

\begin{table}[t]
    \centering
    \caption{%
        Comparing \vampire{} with and without subsumption demodulation
        on SMT-LIB, using \vampire{} in portfolio mode.
    }
    \label{tab:experiment-casc-smtlib}
    \begin{tabular}{l||c|c|c}
        Configuration   & Total     & Solved    & New (SAT+UNSAT)   \\
        \hline
        \vampire{}          & 22,833    & 13,705     & --               \\
        \hline
        \vampire{}, with FSD             & 22,833    & 13,620    & 55 (1+54)         \\
        \vampire{}, with BSD             & 22,833    & 13,632    & 48 (0+48)         \\
        \vampire{}, with FSD and BSD         & 22,833    & 13,607    & 76 (0+76)         \\
    \end{tabular}
\end{table}

Our results are summarized in Tables~\ref{tab:experiment-casc}-\ref{tab:experiment-casc-smtlib}. The
first column of these tables lists the \vampire{} version and
configuration, where \vampire{} refers to
\vampire{} in its portfolio mode (version 4.4). 
Lines 2-4 of these tables use our new \vampire{}, that is our
implementation of subsumption demodulation in \vampire{}.
The column ``Solved'' reports, respectively,
the total number of TPTP and SMT-LIB problems solved
by the considered \vampire{} configurations.
Column ``New'' lists, respectively,
the number of TPTP and SMT-LIB problems solved by
the version with subsumption demodulation but not by the portfolio version of \vampire{}.
This column also indicates in parentheses how many of the solved problems were satisfiable/unsatisfiable.

While in total the portfolio mode of \vampire{} can solve
more problems, we note that this comes at no suprise as the portfolio
mode of \vampire{} is highly tuned using the existing \vampire{}
options. 
In our experiments, we were interested to see whether subsumption
demodulation in \vampire{} can solve problems that cannot be solved by
the  portfolio mode of \vampire{}. The columns ``New''  of
Tables~\ref{tab:experiment-casc}-\ref{tab:experiment-casc-smtlib} give
practical evidence  of the  impact of subsumption demodulation: there
are 30 new TPTP problems and 76 SMT-LIB problems%
\footnote{\scriptsize The list of these new problems is available at \\ \url{https://gist.github.com/JakobR/605a7b7db0101259052e137ade54b32c}}
that the portfolio version of \vampire{} cannot
solve, but forward and backward subsumption demodulation in \vampire{}
can.

\noindent{\bf New problems solved only by subsumption demodulation.} %
%
Building upon our results from Tables~\ref{tab:experiment-casc}-\ref{tab:experiment-casc-smtlib}, we
analysed how many new problems subsumption
demodulation in \vampire{} can solve when compared to other
state-of-the-art reasoners. To this end, we evaluated our work against
the superposition
provers \eprover{} (version 2.4) and \spass{} (version 3.9),
as well as the SMT solvers \textsc{CVC4} (version 1.7) and {\sc Z3}
(version 4.8.7).
We note however, that when using our 30 new problems from
Table~\ref{tab:experiment-casc}, we could not compare our results
against \textsc{Z3} as \textsc{Z3} does not natively parse TPTP. On
the other hand, when using our 76 new problems from
Table~\ref{tab:experiment-casc-smtlib}, we only compared against
\textsc{CVC4} and \textsc{Z3}, as \eprover{} and \spass{} do not
support the SMT-LIB syntax.

Table~\ref{tab:experiment-casc-remains} summarizes our findings.
First, 11 of our 30 ``new'' TPTP problems can only be solved
using forward and backward subsumption demodulation in \vampire{};
none of the other systems were able solve these problems.

Second, while all our 76 ``new'' SMT-LIB problems
can also be solved by {\sc CVC4} and {\sc Z3} together,
we note that out of these 76 problems there are 10 problems that {\sc CVC4} cannot solve,
and similarly 27 problems that {\sc Z3} cannot solve.

\begin{table}[t]
    \centering
    \caption{%
        Comparing \vampire{} with subsumption demodulation against
        other solvers, using the ``new'' TPTP and SMT-LIB problems of
        Tables~\ref{tab:experiment-casc}-\ref{tab:experiment-casc-smtlib}
        and running \vampire{} in portfolio mode.
    }
    \label{tab:experiment-casc-remains}
    \begin{tabular}{l|c||c}
        Solver/configuration    &  TPTP problems & SMT-LIB problems\\
        \hline
        Baseline: \vampire{}, with FSD and BSD      & 30    & 76    \\
        \hline
        \eprover{} with \texttt{-{}-auto-schedule}  & 14    &  -    \\
        \spass{} (default)                          &  4    &  -    \\
        \spass{} (local contextual rewriting)       &  6    &  -    \\
        \spass{} (subterm contextual rewriting)     &  5    &  -    \\
        \textsc{CVC4} (default)                     &  7    & 66    \\
        \textsc{Z3} (default)                       &  -    & 49    \\
        \hline
        Only solved by \vampire{}, with FSD and BSD & 11    &  0    \\
    \end{tabular}
\end{table}

\noindent{\bf Comparative experiments without AVATAR.}
Finally, we investigated the effect of subsumption demodulation
in \vampire{} without AVATAR~\cite{AVATAR14}.
We used the default mode of~\vampire{}
(that is, without using a portfolio approach)
and turned off the AVATAR setting.
While this configuration solves less problems than the portfolio mode of \vampire{},
so far \vampire{} is the only superposition-based theorem prover implementing
AVATAR.  Hence, evaluating  subsumption demodulation in \vampire{} without
AVATAR is more relevant to other reasoners.
Further, as AVATAR may often split non-unit clauses into unit clauses, it may potentially simulate applications of subsumption demodulation
using demodulation.
Table~\ref{tab:experiment-casc-default-av-off} shows that this is indeed the case:
with both \texttt{fsd} and \texttt{bsd} enabled, subsumption
demodulation in \vampire{} can prove 190 TPTP problems and 173
SMT-LIB examples that the default \vampire{}
without AVATAR cannot solve. Again, the column ``New'' denotes the
number of problems solved by the respective configuration but not by
the default mode of \vampire{} without AVATAR. 

\begin{table}[t]
    \centering
    \caption{%
        Comparing \vampire{} in default mode and without AVATAR,
        with and without subsumption demodulation.
    }
    \label{tab:experiment-casc-default-av-off}
    \begin{tabular}{l||c|c|c||c|c|c}
        ~ & \multicolumn{3}{c||}{TPTP problems} & \multicolumn{3}{c}{SMT-LIB problems} \\
        \hline
        \hline
        Configuration                   & Total     & Solved    & New           & Total     & Solved    & New           \\
        ~                               &           &           & (SAT+UNSAT)   &           &           & (SAT+UNSAT)   \\
        \hline
        \vampire{}                      & 13,924    & 6,601     & --            & 22,833    & 9,608     & --            \\
        \hline
        \vampire{}, with FSD            & 13,924    & 6,539     & 152 (13+139)  & 22,833    & 9,597     & 134 (1+133)   \\
        \vampire{}, with BSD            & 13,924    & 6,471     & 112 (12+100)  & 22,833    & 9,541     & 87 (0+87)     \\
        \vampire{}, with FSD and BSD    & 13,924    & 6,510     & 190 (15+175)  & 22,833    & 9,581     & 173 (1+172)   \\
    \end{tabular}
\end{table}

\nocite{Tange:2018:GNUParallel}

%% file: conclusion.tex
\vspace{-1em}
\section{Conclusion}
\vspace{-1em}
We introduced the simplifying inference rule subsumption demodulation to improve support for
reasoning with conditional equalities in superposition-based
first-order theorem proving.
Subsumption demodulation revises existing machineries of superposition
provers and can therefore be efficiently integrated in superposition
reasoning. Our implementation in \vampire{} shows that subsumption
demodulation solves many new examples that existing provers, including
first-order and SMT solvers, cannot handle. Future work includes the
design of more sophisticated approaches for selecting rewriting
equalities and improving the imperfect filtering of clauses indexes.

{\small
\paragraph{\bf Acknowledgements.}
This work was funded by
the ERC Starting Grant 2014 SYMCAR 639270,
the ERC Proof of Concept Grant 2018 SYMELS 842066,
the Wallenberg Academy Fellowship 2014 TheProSE, and
the Austrian FWF research project W1255-N23.
}